\documentclass[preprint,12pt]{elsarticle}

\usepackage{amssymb}

 \usepackage{amsthm}

\journal{Metamaterials}

\begin{document}

\def \e{\varepsilon}
\def \o{\omega}
\def \Div{\hbox{\rm div}}
\def \be{\begin{equation}}
\def \ee{\end{equation}}

\newtheorem*{Theorem}{Theorem}
\newtheorem{Proposition}{Proposition}
\newtheorem{Lemma}{Lemma}

\begin{frontmatter}



\title{Homogenizing metamaterials, three times}


\author{Didier Felbacq}

\address{University of Montpellier 2 and Institut Universitaire de France\\ Laboratory Charles Coulomb Unit\'e Mixte de Recherche du Centre National de la Recherche Scientifique 5221\\ Montpellier, France}

\begin{abstract}
The homogenization of a metamaterial made of a collection of scatterers periodically disposed is studied from three different points of view. Specifically tools for multiple scattering theory, functional analysis, differential geometry and optimization are used. Detailed numerical results are given and the connections between the different approaches are enlightened. 
\end{abstract}

\begin{keyword}
homogenization theory \sep metamaterials \sep photonic crystals \sep scattering theory

\end{keyword}

\end{frontmatter}


\section{Introduction}
\label{}
Giving a general definition of what homogenization is is difficult, because of the various meanings attached to it. 
The physics at stake is the situation when a wave illuminates a complicated object, generally consisting of a periodic set of scatterers, contained in some domain $\Omega$ and gives rise to a diffracted field $U^s$. Loosely speaking, the homogenization problem consists in identifying {\it homogeneous} constitutive relations such that the same domain $\Omega$ inside which these constitutive relations hold, leads to a diffracted field $U_h^s$ such that $U^s$ and $U^s_h$ are close to each other, in some meaning to be specified. This can be done reasonably only if the wavelength $\lambda$ is larger than the period $d$. This identifies a small parameter $\eta=d/\lambda$.

Having been educated by mathematicians, the definition that I would consider the best one is the following: consider a partial differential equation $P_{\eta}$ with oscillating coefficients $a_{\eta}(x)$. Consider a solution $u_{\eta}$ of the equation: $P_{\eta}(u_{\eta})=f$, where $f$ is some convenient source term. Then the goal of homogenization theory is to find a convenient topology in which $u_{\eta}$ converges to a function $u_0$ satisfying an equation:
$P_0(u_0)=f$. The operator $P_0$ is called the homogenized operator \cite{tartar}. This definition is quite clear and at the end, it leads to results such as: ``when $\eta$ tends to $0$, $u_{\eta}$ tends to $u_0$ in some specific meaning''. The point of being able to specify a convergence is very interesting, in that it gives a clear meaning to the question ``how close are $U^s$ and $U^s_h$ ?''.

In the metamaterials community, it is not rare that the notion of parameters extraction be used as a homogenization scheme \cite{simovski}. In that situation the structure is considered a black box (or rather a black slab!). This pragmatic approach, although it might be useful, cannot be considered a homogenization procedure.

Sometimes the Bloch spectrum is used as well. However, care should be taken because of the following result:
\begin{Proposition}
 Given any isofrequency dispersion curves given implicitly in the form $F(k_x,k_y)=0$, there exists a spatially and temporally dispersive permittivity $\varepsilon(k,\omega)$ reproducing these curves.
 \end{Proposition}
\begin{proof} 
Write: $k^2=k^2+F(k_x,k_y)$ and define: $\varepsilon(k,\omega)=\frac{k^2+F(k_x,k_y)}{k_0^2}$, where $k_0=\o/c$.
\end{proof}
 This shows that given any dispersion curves, it is always possible to reproduce it by using a spatially dispersive permittivity. However, nothing can be said on the complete electromagnetic field inside the medium: the Bloch diagram only accounts for the plane wave part of the field. The introduction of spatial dispersion should only be made with great care \cite{alex}.
\begin{figure}
   \begin{center}
   \includegraphics[height=6cm]{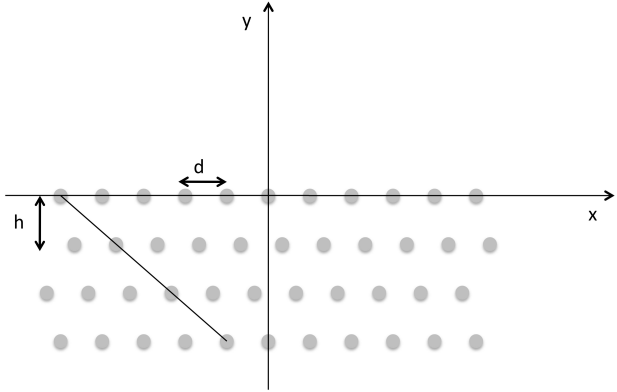}
   \end{center}
   \caption
   { \label{fig1} Sketch of the structure under study. It is made of a stack of gratings, each consisting of periodically disposed scatterers.}
\end{figure} 

In the following, the structure considered as a model problem is a periodic set of 2D scatterers (cf. fig. (\ref{fig1})). The medium is infinite in the $x$ direction and it is made out of a stack of basic layers made of an infinite number of rods periodically disposed at points $x=n d$. 

In the following, three different points of view are given on the homogenization of this structure: multiple scattering, double-scale, micro-local.

\section{Multiple scattering homogenization}
The approach proposed here is reminiscent of that in \cite{alu}, although it is rather different because here no averaged field is defined. It would be interesting to make a connection between these two approaches. \\
Our point is to describe the electromagnetic behavior of one line of scatterer (a grating) alone. Each scatterer ``n'' is characterized by a scattering coefficient $s_0$. When it is illuminated by an incident field $E^i$, it gives rise to a field $u^s(r)=b H_0(k_0|r|)$ where $b=s_0 E^i(n)$. For the infinite set of scatterers, this gives a diffracted field that reads:
\be
u^s=\sum_n b_n H_0(k_0|r-r_n|)
\ee
the field diffracted by the n$^{th}$ scatterer is obtained by writing that it is the response to the incident field $E^i$ and 
the field difffracted by the other scatterers:
\be
b_n=s_0\left( E^i(n)+\sum_{m \neq n} b_m H_0(k_0|m-n|d) \right)
\ee
If the incident field is pseudo-periodic, i.e. $E^i(n)=e^{ikd} E^i (n-1)$, then it holds:
\be
b_0=\frac{ E^i(0)}{s_0^{-1}-\sum_{m \neq 0}  e^{ikm}H_0(k_0|m|d)}
\ee
the series that enters this relation can be written:
\begin{eqnarray*}
\sum_{m \neq 0}  e^{ikmd}H_0(k_0|m|d)=-1-\frac{2i}{\pi}\gamma+\frac{2i}{\pi}\ln\left(\frac{4\pi}{k_0 d}\right)+\frac{2}{d\beta_0}\\
+\frac{2}{d}\sum_{n>0} \left( \frac{1}{\beta_n}+\frac{1}{\beta_{-n}}-\frac{d}{i\pi |n|}\right)
\end{eqnarray*}
An asymptotic analysis of this series \cite{bloch} allows to write the following expansion :
\begin{Lemma}
\be
\frac{2}{d}\sum_{n>0} \left( \frac{1}{\beta_n}+\frac{1}{\beta_{-n}}-\frac{d}{i\pi |n|}\right)
=O[(k_0d)^2]
\ee
\end{Lemma}
Above the grating, the propagative part of the electric field reads as:
$E^+(y)=e^{-i\beta_0 y}+r(k_0,\beta_0) e^{i\beta_0 y}$ and below it reads $E^-(y)=t(k_0,\beta_0) e^{-i\beta_0 y}$ where: $r(k_0,\beta_0)=\frac{2}{\beta_0 d}b_0$ and $t=1+r$

Energy conservation implies that : $|r|^2+|1+r|^2=1$ and therefore the following representation holds:

\begin{Proposition}
 There exists a real function $\chi(k_0,\beta_0)$ such that: \[r(k_0,\beta_0)=\frac{-1}{1+i\chi(k_0,\beta_0)}\]
\end{Proposition}
\begin{proof} energy conservation shows that: $\frac{\Re{(r)}}{|r|^2}=-1$, the theorem follows by defining:
$\chi=\frac{\Im{(r)}}{|r|^2}$.
\end{proof}
\begin{Proposition}
For the propagative part of the field, the grating is equivalent to an infinitely thin slab whose transfer matrix is:
\be
T_g=\left(
\begin{array}{cc}
1 & 0 \\ \frac{1}{L} & 1
\end{array}
\right)
\ee
where $L=\chi/2\beta_0$. When the ratio $a/\lambda$ is very small, it holds: $L\sim\frac{d}{2\pi} \ln\left(\frac{d}{2\pi a}\right)$
\end{Proposition}
\begin{proof}
The electric field is continuous, the matrix $T_g$ is obtained by computing the jump of the normal derivative of the field: $\partial_y E^+(0)-\partial_y E^-(0)=2i\beta_0 r$. Using proposition 1, we get: $E(0)=1+r=\frac{i\chi}{1+i\chi}=-i\chi\,r$. Hence: $2i\beta_0 r=-\frac{2\beta_0}{\chi}E(0)$.
\end{proof}
A layer of the structure can thus be described by this slab surrounded by homogeneous layers of height $h/2$, the transfer matrix of this basic sandwich structure (cf. fig. (\ref{fig2})) being: $T=T(h/2)T_g T(h/2)$, where :
\be
T(h/2)=\left(
\begin{array}{cc}
\cos(\beta_0 h/2) & \sin(\beta_0 h/2)/\beta_0 \\
-\beta_0 \sin(\beta_0 h/2)  & \cos(\beta_0 h/2)
\end{array}
\right)
\ee
A simple calculation shows that:
\be
T=\left(\begin{array}{cc} 
\cos\!\left(\beta_0\, h\right) + \frac{\sin\!\left(\beta_0\, h\right)}{2\, L\, \beta_0} & \frac{{\sin\!\left(\frac{\beta_0\, h}{2}\right)}^2 }{L\, \beta_0^2}+\frac{\sin\!\left(\beta_0\, h\right)}{\beta_0}\\ 
\frac{{\cos\!\left(\frac{\beta_0\, h}{2}\right)}^2}{L} - \beta_0\, \sin\!\left(\beta_0\, h\right) & \cos\!\left(\beta_0\, h\right) + \frac{\sin\!\left(\beta_0\, h\right)}{2\, L\, \beta_0} 
\end{array}\right)
\ee
The transfer matrix is known to be an unstable object for numerical purposes. The preferred quantity is the scattering matrix \cite{brahim}. Therefore, the homogenization is now performed numerically by searching for a permittivity $\e_{\rm eff}(k_0,\beta_0)$ such that the scattering matrix of a homogeneous slab with this permittivity fits that of the sandwich structure. Specifically, the cost function is:
\begin{equation}\label{opt}
J[\e_{\rm eff}]=|r_{g}-r_{\rm eff}|^2+|t_g-t_{\rm eff}|^2
\end{equation}
where $r_g$ and $t_g$ are the reflection and transmission coefficients of the sandwich structure and $r_{\rm eff}, t_{\rm eff}$ that of the homogeneous slab.
\begin{figure}
   \begin{center}
   \includegraphics[height=6cm]{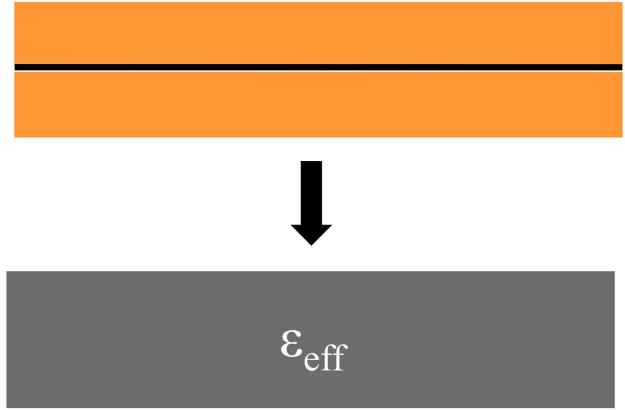}
   \end{center}
   \caption
   { \label{fig2} The homogenization of the sandwich structure. The horizontal black line represents the grating.}
\end{figure} 
\subsection{Numerical applications}
First the scatterer in the basic cell is a dielectric rod with permittivity $\e_r=9$ and radius $a/d=1/2$, i.e. the rods are touching. In that case, the low-frequency behavior is well-known: the sandwich structure is equivalent to a homogeneous slab with permittivity: $\e_{\rm lw}=1+(\e_r-1)\pi \left( \frac{a}{d}\right)^2$. 
\begin{figure}
   \begin{center}
   \includegraphics[height=6cm]{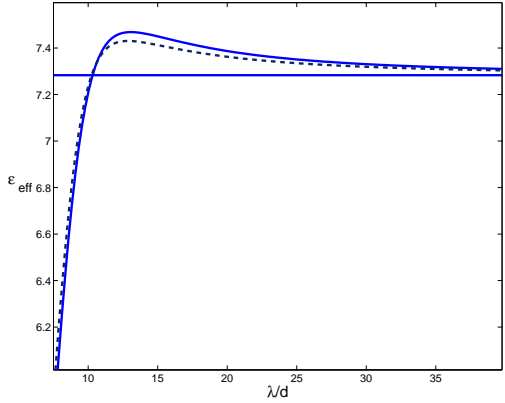}
   \end{center}
   \caption
   { \label{fig3} The reflection spectrum for the homogenized structure (dashed) and the sandwich structure (continuous) .}
\end{figure} 

This serves to test the proposed approach and also to evaluate the spatial dispersion effect in the medium. 
The reflection spectra for a plane wave in normal incidence and for both the sandwich structure and the homogeneous slab are given in fig. (\ref{fig3}). It can be shown that the homogenization works very well for $\lambda/d > 15$.
The resulting homogenized permittivity, depending on both the frequency and the horizontal Bloch vector, is given in fig.(\ref{fig4}). The averaged value is $\e_{\rm lw}=7.28318 \pm 10^{-5}$. The value obtained numerically for $\lambda/d=1000$ is: $\e_{\rm eff}=7.2832$. The homogenized permittivity was calculated for two angles of incidence $0$ and $\pi/4$: it can be seen that the effects of spatial dispersion are quite small.

\begin{figure}
   \begin{center}
   \includegraphics[height=6cm]{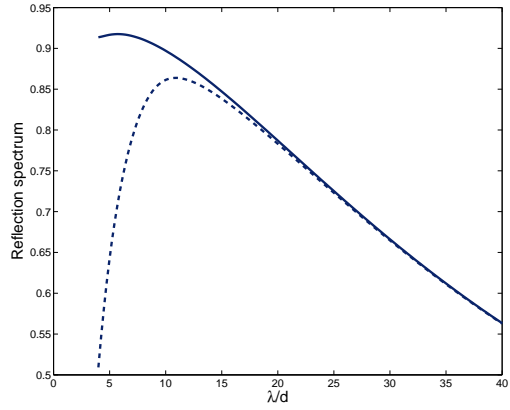}
   \end{center}
   \caption
   { \label{fig4} The effective permittivity obtained by using the cost function (\ref{opt}). The continuous curve corresponds to $\beta_0=0$ and the dashed one to $\beta_0=k_0 \sqrt{2}/2$. The horizontal line corresponds to the averaged permittivity $\e_{\rm lw}$. }
\end{figure} 

Second, we choose a resonant scatterer with a scattering coefficient of the form 
\begin{equation}\label{polo}
s_0=\frac{1}{g(\o)}\frac{\o-\o_z}{\o-\o_p}, 
\end{equation}
where $g(\o)$ is regular and satisfy $g \rightarrow 0$ as $\o \rightarrow 0$. This pole-and-zero form is quite a common one. The sandwich structure air-grating-air can be replaced by a slab with permittivity $\e_{\rm eff}(\omega,\beta_0)$ obtained from the optimisation procedure described above. 
\begin{figure}
   \begin{center}
   \includegraphics[height=6cm]{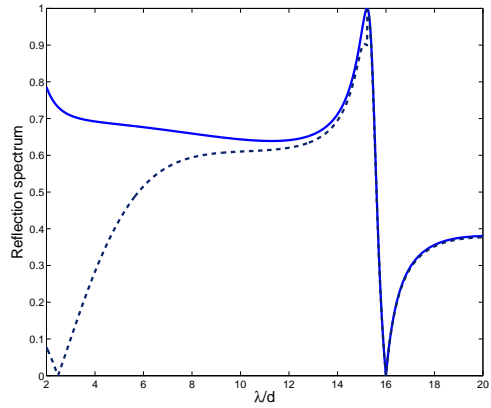}
   \end{center}
   \caption
   { \label{fig5} The effective permittivity obtained by using the cost function (\ref{opt}) for a grating of resonant scatterers with scattering coefficient given by (\ref{polo}).}
\end{figure} 
The reflection spectrum is given in fig. (\ref{fig5}), the homogenization procedure works very well for $\lambda/d >12$.
\begin{figure}[h!]
   \begin{center}
   \includegraphics[height=6cm]{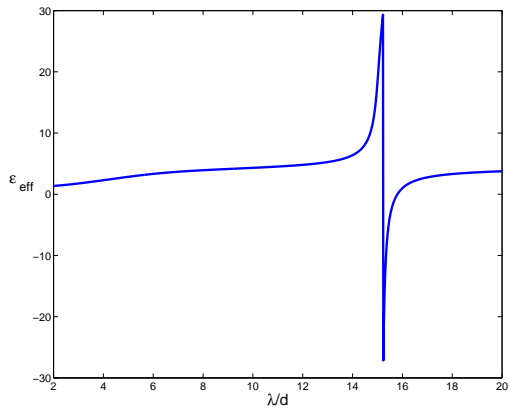}
   \end{center}
   \caption
   { \label{fig6} The effective permittivity obtained by using the cost function (\ref{opt}) .}
\end{figure} 
The corresponding homogenized permittivity is given in fig.(\ref{fig6}). The spatial dispersion effects are negligible in that situation. Interestingly, it seems that resonant structures can be homogenized at smaller ratio $\lambda/d$ than non resonant ones.
\section{Pseudo-differential homogenization}
In this section, a homogenization approach inspired by \cite{silveirinha} is given. For the sake of simplicity, the theory is specialized to the Helmholtz equation: $\Delta u+k_0^2 \e u=f$, where $\e$ is the dielectric function of the structure in fig.(\ref{fig1}). Let us define by $T$ the operator such that: $u=T(f)$. It is a pseudo-differential operator \cite{pseudo} whose symbol is denoted $a(r,k)$. The following form of $u$ holds:
\begin{equation}
u(r)=\int a(r,k) \hat{f}(k) e^{ik \cdot r} dk
\end{equation}
The symbol is the response of the system for a plane wave source. It can be shown that it has an expansion in terms of the Bloch spectrum \cite{prepa}:
\begin{equation}
a(r,k)=\sum_p  \frac{|\phi_p(r,k)|^2}{E_p(k)-k_0^2}
\end{equation}
where $\phi_p$ is such that: $\phi_p e^{-ikr}$ is periodic and $-\e^{-1}\Delta \phi_p=E_p \phi_p$.\\
Let us now consider the periodic medium with basic cell $Y$. The polarization field is: $P_{\eta}=\e_0(\e-1) u_{\eta}$. 
Define the average field through: $u_{\rm eff}=\int_Y a_{\eta}(r,k) e^{ik \cdot r} dr$ and the average of the polarization field: $$P_{\rm eff}=\e_0\int_Y (\e-1) a_{\eta}(r,k) e^{ik \cdot r} dr=\e_0(\e-1)\int_D  a_{\eta}(r,k) e^{ik \cdot r} dr$$
The effective permittivity is then $\e_{\rm eff}(k,\o)=1+\frac{P_{\rm eff}}{\e_0 u_{\rm eff}}$, that is:
\begin{equation}
\e_{\rm eff}(k,\o)=1+(\e-1)\frac{\int_D  a_{\eta}(r,k) e^{ik \cdot r} dr}{\int_Y  a_{\eta}(r,k) e^{ik \cdot r} dr}
\end{equation}
It is interesting to note that the expression above is regular whatever $k$, despite the fact that the symbol has poles. In the very low frequency domain, one obviously has:
\begin{equation}
\e_{\rm eff}(k,\o)=1+(\e-1)\frac{|D|}{|Y|}
\end{equation}

This approach could probably be made better by using a multiple scale expansion of the symbol in the form: $$a_{\eta}(r,k)=a_0(r,r/\eta,k)+\eta a_1(r,r/\eta,k)+...$$. Besides, it would also be interesting to recover the classical homogenization results for the case of a magnetic field linearly polarized along the wires. Work is in progress in that direction.

\section{Multiple scale homogenization}
This final approach is largely used in the mathematical community \cite{kozlov,bensoussan}, also it is practically ignored by physicists \cite{tretyadov}. By its very definition, it is the only one that can give a clear meaning to the notion of convergence of the fields. I will only give a sketch in a well-known situation, but I will give a new framework by using differential forms \cite{tartar}, because the structure of the theory is very nice then.
Maxwell equations  are written by using differential forms:
\be
d E^{\eta}=i \o \mu_0 H^{\eta} ,\, d H^{\eta}=-i\o \e_{\eta} \star E^{\eta}
\ee
where $d$ denotes the exterior derivative, $E^{\eta}=E_n^{\eta}dx^n$ and $H^{\eta}=H_n^{\eta}dx^n,\, n=1,2,3$
are 1-forms (obtained from the usual vector fields by means of the $\flat$ operator and a flat metric \cite{steenrod}) and $\star$ stands for the Hodge operator.
The starting point is to assume the following expansions for $E^{\eta}$ and $H^{\eta}$:
\begin{equation}\label{2scale}
\begin{array}{c}
E^{\eta}(x)=E^0(x,x/\eta)+\eta \, E^1(x,x/\eta)+\eta^2 E^2(x,x/\eta)+...\,,\\
 H^{\eta}(x)=H^0(x,x/\eta)+\eta \, H^1(x,x/\eta)+\eta^2 H^2(x,x/\eta)+...
 \end{array}
\end{equation}
where the fields $E^0,H^0,E^1, H^1...$ depend upon two sets of variables $(x,y)\in \mathbb{R}^3 \times Y$.

This can be justified in the framework of double-scale convergence \cite{allaire,nguetseng} but this would take us to far.

There is a nice mathematical structure linked to the multiscale expansion. The projection:
\be
\pi: E_0(x,y) \longrightarrow \int_Y E_0(x,y) dy
\ee
defines a fiber bundle. Where the fiber is the cotangent bundle to the flat torus: $F=T^{\ast}(\mathbb{R}^2/\mathbb{Z}^2)$ and the base is the cotangent bundle to the ambient space :${\cal B}=T^{\ast}(\mathbb{R}^2)$.Its trivialization is as follows:
\be
\pi^{-1}(U) \sim F \times {\cal B}
\ee
where $U$ is an open neighborhood in the base. The $y$ variable plays the role of a hidden variable that accounts for the microscopic behavior inside the basic cell.

The expansion (\ref{2scale}) induces the following splitting of operator $d$: $d=d_x+\frac{1}{\eta} d_y$. We arrive directly at the systems satisfied by the microscopic fields, holding on $Y$:
\begin{equation}\label{micro}
\left\{
\begin{array}{l}
d_y E_0=0 \\
\delta_y( \e E_0)=0 
\end{array}
\right .
\left\{
\begin{array}{l}
d_y B_0=0\\
\delta_y B_0=0
\end{array}
\right .
\end{equation}
where $\delta=-\star d \star$ is the exterior co-derivative.
The first system involves the electric field alone and is of purely electrostatic nature. Because the limit field only depends microscopically on $(y_1,y_2)$, one immediately obtains: $E^0_3(x,y)=E_3(x)$. Moreover the transverse microscopic electric field  $E^0_{\bot}=(E^0_1,E^0_2)$ is exact on the basic cell $Y$. This implies that $E^0_{\bot}(x,\cdot)$ reads as:
\be
E^0_{\bot}(x,\cdot)=E_{\bot}(x)+d_y \varphi(\cdot)
\ee
and therefore the effective electric field $E_{\bot}(x)=(E_1(x),E_2(x))$ belongs to the first de Rham cohomology space $H^1(Y)$, a space isomorphic to $\mathbb{R}^2$. The complete solving of the miscropic system is done by a linear decomposition: $\varphi=E_1 \varphi_1+E_2 \varphi_2$, where the functions $\varphi_j, j=1,2$ satisfy:
\begin{equation}
\begin{array}{c}
\Delta_y \varphi_j=0 \hbox{ on } Y \setminus D \\
\partial_n \varphi_j=-n_j \hbox{ on } \partial D
\end{array}
\end{equation}
This leads to the linear relation: $E_0={\cal E}(y) E(x)$, where:
\begin{equation}
{\cal E}(y) =\left(
\begin{array}{ccc}
1+\partial_1 \varphi_1 & \partial_1 \varphi_2  & 0\\
\partial_2 \varphi_1  &  1+\partial_2 \varphi_2 & 0 \\
0 & 0 & 1
\end{array}
\right)
\end{equation}

The second system shows that the magnetic field is both exact and co-exact, implying that $B_0$ does not depend on $y$. We denote: $B_0(x,y)=B(x)$.
Finally, the macroscopic equations read as:
\begin{equation}
d E^0=i\o \mu_0 H^0 ,\,
d H^0=-i\o \e_0\e E^0
\end{equation}
After averaging on $Y$, one obtains:
\begin{equation}
d E=i\o \mu_0 H ,\,
d H=-i\o \e_0 \e_{\rm eff} E
\end{equation}
where $\e_{\rm eff}=\int _Y \e(y) {\cal E}(y)\,dy$.
The usual anisotropic permittivity tensor is found. Generalizations and details can be found in \cite{jnp1,quasi}. The method can be extended to deal with resonant structures \cite{cras1} and obtain a homogenization result for higher bands that the first one, also it is sometimes wrongly believed that this approach can only deal with quasistatic problems. It is interesting to note that the degree of the forms involved can be a clue to the definition of the averaged field \cite{cras2}.

\section{Conclusion}
We have described three different approaches to the homogenization of a two dimensional dielectric metamaterial. The micro-local approach is less developed than the other but still seems quite interesting, and in need of mathematical development. Our preferred one still is the multiple scale approach because it can deal with the notion of convergence and the boundary conditions naturally. However, it is also in need of mathematical development in order to take spatial dispersion into account.
\vskip 1cm
{\bf Acknowledgments}\\
The financial support of the Agence Nationale de la Recherche through grant 060954 OPTRANS is acknowledged. D. Felbacq is a member of the Institut Universitaire de France.

\end{document}